\documentclass[a4paper,11pt]{article}

\usepackage[utf8]{inputenc}
\usepackage[margin=3cm]{geometry}
\usepackage{xparse}
\usepackage{amsmath,amssymb,amsthm}
\usepackage{mathtools}
\usepackage{bbm}
\usepackage{fourier}
\usepackage{xcolor}
\usepackage[numbers]{natbib}
\usepackage{todonotes}
\usepackage{enumitem}
\usepackage[pdfpagelabels]{hyperref}
\hypersetup{colorlinks,citecolor=black,linkcolor=black,urlcolor=black}

\newtheorem{theorem}{Theorem}
\newtheorem{corollary}[theorem]{Corollary}
\newtheorem{lemma}[theorem]{Lemma}
\newtheorem{proposition}[theorem]{Proposition}

\DeclareDocumentCommand\setdef{mo}{\left\{#1\IfNoValueTF{#2}{}{ : #2}\right\}}
\DeclareDocumentCommand\scalprod{mm}{\left<#1,#2\right>}
\DeclareDocumentCommand\R{}{\mathbb{R}}

\DeclareDocumentCommand\orderO{m}{O\left(#1\right)}
\DeclareDocumentCommand\orderOmega{m}{\Omega\left(#1\right)}

\DeclareDocumentCommand\cplxNP{}{\mathsf{NP}}
\DeclareDocumentCommand\zerovec{o}{\IfNoValueTF{#1}{\mathbb{O}}{\mathbb{O}_{#1}}}

\DeclareDocumentCommand\cone{o}{\operatorname{cone}\IfValueTF{#1}{\left(#1\right)}{}}
\DeclareDocumentCommand\conv{o}{\operatorname{conv}\IfValueTF{#1}{\left(#1\right)}{}}
\DeclareDocumentCommand\xc{o}{\operatorname{xc}\IfValueTF{#1}{\left(#1\right)}{}}
\DeclareDocumentCommand\unitvec{m}{\mathbbm{e}_{#1}}
\DeclareDocumentCommand\dominant{m}{\ensuremath{#1^{\uparrow}}}
\DeclareDocumentCommand\proj{oo}{\IfValueTF{#1}{\operatorname{proj}{}_{#1}}{\operatorname{proj}{}}\IfValueTF{#2}{\left(#2\right)}{}}
\DeclareDocumentCommand\PolyPerfMatch{m}{P_{\text{pmatch}}(#1)}

\title{Extended Formulations for Radial Cones}
\author{Matthias Walter\footnote{RWTH Aachen University, Germany; \url{walter@or.rwth-aachen.de}} \and Stefan Weltge\footnote{Technical University of Munich, Germany; \url{weltge@tum.de}}}

\ProvideDocumentCommand\keywords{m}{\small \textbf{Keywords --- } #1}

\begin{document}

\maketitle

\begin{abstract}
  This paper studies extended formulations for radial cones at vertices of polyhedra, where the radial cone of a polyhedron $ P $ at a vertex $ v \in P $ is the polyhedron defined by the constraints of $ P $ that are active at $ v $.
  Given an extended formulation for $ P $, it is easy to obtain an extended formulation of comparable size for each its radial cones.
  On the contrary, it is possible that radial cones of $ P $ admit much smaller extended formulations than $ P $ itself.

  A prominent example of this type is the perfect-matching polytope, which cannot be described by subexponential-size extended formulations (Rothvo\ss{} 2014).
  However, Ventura \& Eisenbrand~(2003) showed that its radial cones can be described by polynomial-size extended formulations.
  Moreover, they generalized their construction to $ V $-join polyhedra.
  In the same paper, the authors asked whether the same holds for the odd-cut polyhedron, the blocker of the $ V $-join polyhedron.

  We answer this question negatively.
  Precisely, we show that radial cones of odd-cut polyhedra cannot be described by subexponential-size extended formulations.
  To obtain our result, for a polyhedron $ P $ of blocking type, we establish a general relationship between its radial cones and certain faces of the blocker of $ P $.
\end{abstract}

\keywords{radial cones; extension complexity; matching polytope; odd-cut polyhedron}

\section{Introduction}
The concept of extended formulations is an important technique in discrete optimization that allows for replacing the inequality description of some linear program by another inequality description of preferably smaller size using auxiliary variables.
Geometrically, given a polyhedron $ P \subseteq \R^p $ one searches for a polyhedron $ Q \subseteq \R^q $ together with a linear map $ \pi : \R^q \to \R^p $ such that $ \pi(Q) = P $.
The pair $ (Q, \pi) $ is called a \emph{linear extension} of $ P $ whose \emph{size} is the number of facets of $ Q $.

There are several polyhedra associated to classic combinatorial optimization problems having a large number of facets but admitting linear extensions of small size (polynomial in their dimension).
Prominent examples are the spanning tree polytope~\cite{Wong80,Martin91}, the subtour elimination polytope~\cite{Wong80}, and the cut dominant~\cite[\S 4.2]{ConfortiCZ13}.
On the other hand, the seminal work of Fiorini et al.~\cite{FioriniMPTW12} has shown that such descriptions do not exist for many polytopes associated to hard problems, including the cut polytope or the travelling salesman polytope.
Surprisingly, the same is true even for the perfect-matching polytope, a very well-understood polytope over which linear functions can be optimized in polynomial time~\cite{Edmonds65a}.
In fact, Rothvoß~\cite{Rothvoss14} proved that every linear extension of the perfect-matching polytope $ \PolyPerfMatch{n} $ of the complete graph $ K_n = (V_n,E_n) $ on $ n $ nodes has size $ 2^{\orderOmega{n}} $.

Thus, in terms of sizes of linear extensions, the perfect matching polytope appears as complicated as certain polytopes associated to hard problems.
Ventura \& Eisenbrand~\cite{VenturaE03} showed that this situation changes if one aims for \emph{local} descriptions:
Given a vertex $ v $ of $ \PolyPerfMatch{n} $, they showed that the polyhedron defined by only those constraints of $ \PolyPerfMatch{n} $ that are active at $ v $, the \emph{radial cone} at $ v $, has a linear extension of size $ O(n^3) $.

Note that such formulations can be used to efficiently test whether a given vertex is optimal with respect to a given linear function.
For linear 0/1-optimization problems, efficient routines for such local checks are usually enough to obtain an efficient algorithm for the actual optimization problem, see~\cite{SchulzWZ95,Schulz09}.
Thus, the work in~\cite{VenturaE03} yields another proof that the weighted matching problem can be solved in polynomial time.
However, this also suggests that such descriptions do not exist for polytopes associated to hard problems, which separates matching from harder optimization problems.

Furthermore, Ventura \& Eisenbrand generalized their construction to the $ V_n $-join polyhedron of $ K_n $ (which contains $ \PolyPerfMatch{n} $ as a face), showing that its radial cones also admit linear extensions of size $ \orderO{n^3} $.
In the same paper, the authors asked whether the same holds for the odd-cut polyhedron, which is the blocker of the $ V_n $-join polyhedron and hence closely related.\footnote{Precise definitions of all relevant terms used in the introduction will be given later.}

\paragraph{Our results.}

\begin{enumerate}[leftmargin=0em]
  \item The main purpose of this work is to answer their question negatively by showing the following result.
  \begin{theorem}
    \label{TheoremMainIntro}
    There exists a constant $ c > 0 $ such that for every even $ n $, the radial cones of the odd-cut polyhedron of $ K_n $ cannot be described by linear extensions of size less than $ 2^{c n} $.
  \end{theorem}

  \item To obtain our result, for a polyhedron $ P $ of blocking type, we establish a general relationship between its radial cones and certain faces of the blocker of $ P $.
  In the case of the odd-cut polyhedron, we show that its radial cones correspond to certain faces of the $ V_n $-join polyhedron that can be shown to require large linear extensions using Rothvoß' result.
  Analogously, it turns out that radial cones of the $ V_n $-join polyhedron correspond to certain faces of the odd-cut polyhedron, which can be easily described by linear extensions of size $ \orderO{n^3} $.
  This allows us to give an alternative proof of the result by Ventura \& Eisenbrand.

  \item We complement our results by observing that radial cones of polytopes associated to most classical hard optimization problems indeed do not admit polynomial-size extended formulations in general.
\end{enumerate}

\paragraph{Outline.}

The paper is structured as follows.
In Section~\ref{SectionOverview}, we will introduce the relevant concepts and derive straight-forward results on extension complexities of radial cones.
Using elementary properties of blocking polyhedra, we will derive a structural relationship between radial cones and certain faces of the blocker in Section~\ref{SectionBlocking}.
Using these insights, our main result is proved in Section~\ref{SectionMain}, where we also provide an alternative proof of the result by Ventura and Eisenbrand.
Finally, an upper bound that complements our main result is provided in the appendix.

\paragraph{Acknowledgements.}
We would like to thank Robert Weismantel for inviting the first author to ETH Zürich, where parts of this research were carried out.

\DeclareDocumentCommand\radialcone{mm}{K_{#1}(#2)}

\section{Overview}
\label{SectionOverview}
Recall that, for a polyhedron $ P $ and a point $ v \in P $, we are interested in describing the \emph{radial cone} $ \radialcone{P}{v} $, which is the polyhedron defined by all inequalities that are valid for $ P $ and satisfied with equality by $ v $.\footnote{Technically, by our definition, $ \radialcone{P}{v} $ is not necessarily a cone. In fact, standard definitions of the radial cone (or the \emph{cone of feasible directions}) differ from ours in a translation by the vector $ -v $, see, e.g.,~\cite[\S~2.2]{Ruszczynski06}. However, the one given here will be more convenient for us.}
Thus, given an inequality description of $ P $, the radial cone is simply defined by dropping some of the inequalities.
Note that a polyhedron arising from $ P $ by deleting an arbitrary subset of inequalities might require much larger linear extensions than $ P $ does.
However, radial cones arise in a very structured way, which allows us to carry over linear extensions for $ P $.
This might become clear by observing that
\[
  \radialcone{P}{v} = \cone(P - v) + v,
\]
where $ \cone(X) \coloneqq \setdef{ \lambda x }[ \lambda \ge 0, \, x \in X ] $.
Let us formalize the previous claim and other basic observations in the following proposition.
To this end, we make use of the (linear) \emph{extension complexity} $ \xc(P) $ of a polyhedron $ P $, which is defined as the smallest size of any linear extension of $ P $.
\begin{proposition}
  \label{TheoremExtensionBasics}
  Let $P \subseteq \R^n$ be a polyhedron and $v \in P$.
  \begin{enumerate}[label=(\roman*)]
  \item
    \label{TheoremExtensionBasicsRadialCone}
    $\xc(\radialcone{P}{v}) \leq \xc(P) + 1$.
  \item
    \label{TheoremExtensionBasicsFace}
    Every face $ F $ of $ P $ satisfies $\xc(F) \leq \xc(P)$.
  \item
    \label{TheoremExtensionBasicsRadialConeFace}
    Every face $ F $ of $ P $ with $ v \in F $ satisfies $\xc(\radialcone{F}{v}) \leq \xc(\radialcone{P}{v})$.
  \item
    \label{TheoremExtensionBasicsProjection}
    For every linear map $\pi : \R^p \to \R^d$, we have $\xc(\pi(P)) \leq \xc(P)$
    and $\xc(\radialcone{\pi(P)}{\pi(v)}) \leq \xc(\radialcone{P}{v})$.
  \end{enumerate}
\end{proposition}

\begin{proof}
  To see~\ref{TheoremExtensionBasicsRadialCone},
  let $P = \setdef{ x \in \R^p }[ \exists y \in \R^d : Ax + By \leq c ]$, with matrices $A \in \R^{m \times n}$, $B \in \R^{m \times d}$ and $c \in \R^m$ be an extended formulation of $P$ with $m = \xc(P)$.
  It is easy to see that
$$\radialcone{P}{v} = \setdef{ x \in \R^p }[ \exists y \in \R^d, \exists \mu \geq 0 : A(x-v) + By \leq c \mu ],$$
  which concludes the proof of this part.

  Let $ F $ be a face of $ P $ and let $ H $ be a corresponding supporting hyperplane, i.e., $F = P \cap H$.
  Since $ H $ is described by an equation,~\ref{TheoremExtensionBasicsFace} follows.
  Moreover, $\radialcone{F}{v} = \radialcone{P}{v} \cap H$, i.e., the radial cone of $ F $ at $ v $ is a face of the radial cone of $ P $ at $ v $.
  Application of~\ref{TheoremExtensionBasicsFace} yields~\ref{TheoremExtensionBasicsRadialConeFace}.

  The first statement of~\ref{TheoremExtensionBasicsProjection} follows by concatenating the projection map of a minimum-size extension of $ P $ with $ \pi $.
  To prove the second statement, we will show that $\pi(\radialcone{P}{v}) = \radialcone{\pi(P)}{\pi(v)}$.
  By translating $P$ to $P - v$ (and by keeping $\pi$, also translating $\pi(P)$ to $\pi(P) - \pi(v)$), this is equivalent to showing $ \pi( \cone(P) ) = \cone(\pi(P))$ for $\zerovec \in P$.
  Clearly, the last statement holds by linearity of $\pi$.
\end{proof}

Notice that one can get rid of the ``$+1$'' in Proposition~\ref{TheoremExtensionBasics}~\ref{TheoremExtensionBasicsRadialCone} by projecting the radial cone of a minimum-size extension of $P$ at any preimage of $v$.
This makes the proof slightly longer, and we decided to present the simpler proof above.

\DeclareDocumentCommand\PolyCut{m}{P_{\text{CUT}}(#1)}%

On the one hand, Proposition~\ref{TheoremExtensionBasics}~\ref{TheoremExtensionBasicsRadialCone} shows that radial cones of polyhedra admitting small extensions, e.g., the ones mentioned in the introduction, also have a small extension complexities.
On the other hand, the last two statements of the proposition can be used to derive lower bounds on extension complexities of radial cones of polytopes related to many $\cplxNP$-hard problems.

\paragraph{Radial cones of polytopes associated to hard problems.}
Consider the cut polytope $ \PolyCut{n} \in \R^{E_n} $ of the complete graph $ K_n = (V_n,E_n) $ defined as the convex hull of characteristic vectors of cuts (in the edge space) in $ K_n $.
Braun et~al.\ proved (see Proposition~3 in~\cite{BraunFPS12}) that $ \cone(\PolyCut{n}) $ has extension complexity at least $ 2^{\orderOmega{n}} $.
Note that $ \cone(\PolyCut{n}) $ is the radial cone of $ \PolyCut{n} $ at the vertex corresponding to the empty cut. 
Furthermore, it has been shown that several polytopes associated to other $\cplxNP$-hard problems have faces that can be projected onto cut polytopes by (affine) linear maps.
Examples are certain stable-set polytopes and traveling-salesman polytopes~\cite{FioriniMPTW12}, certain knapsack polytopes~\cite{AvisT13,PokuttaV13} and 3d-matching polytopes (see~\cite{AvisT13}).
Consider any such a polytope $P(n)$ and let $F(n)$ be a face that projects to $\PolyCut{n}$.
Clearly, $F(n)$ must have a vertex $v_n$ whose projection is the vertex $\zerovec$ of $\PolyCut{n}$.
By Proposition~\ref{TheoremExtensionBasics}~\ref{TheoremExtensionBasicsRadialConeFace} and~\ref{TheoremExtensionBasicsProjection},
the extension complexity of the radial cone of $P(n)$ at $v_n$ is greater than or equal to the extension complexity of the radial cone of $\PolyCut{n}$ at $\zerovec$.
Hence, for such polytopes, we obtain super-polynomial lower bounds on extension complexities of some of their radial cones.

\DeclareDocumentCommand\PolyTJoin{om}{\ensuremath{P_{\IfValueTF{#1}{#1}{T}\text{-join}}(#2)}}%
\DeclareDocumentCommand\PolyTCut{om}{\ensuremath{P_{\IfValueTF{#1}{#1}{T}\text{-cut}}(#2)}}%

\paragraph{Polyhedra associated to matchings, $ T $-joins, and $ T $-cuts.}
Throughout the paper, let $ T \subseteq V_n $ be a node set of even cardinality.
A \emph{$ T $-join} is a subset $ J \subseteq E_n $ of edges such that a node $ v \in V_n $ has odd degree in the subgraph $ (V_n,J) $ if and only if $ v \in T $.
A \emph{$ T $-cut} is a subset $ C \subseteq E_n $ of edges such that $ C = \delta(S) := \setdef{\{v,w\} \in E_n}[v \in S, \, w \notin S] $ holds for some $ S \subseteq V_n $ for which $ |S \cap T| $ is odd.
The $ V_n $-cuts are also known as \emph{odd cuts}.
The perfect-matching polytope $ \PolyPerfMatch{n} $, $ T $-join-polytope $ \PolyTJoin{n} $ and $ T $-cut polytope $ \PolyTCut{n} $ are defined as the convex hulls of characteristic vectors of all perfect matchings, $ T $-joins and $ T $-cuts of $K_n$, respectively.
The (weighted) minimization problem for $ T $-cuts is $ \cplxNP $-hard for arbitrary objective functions, but can be solved in polynomial time for nonnegative ones~\cite{PadbergR82}.
For this reason we focus on the \emph{dominant} of the $ T $-cut polytope, defined as $ \dominant{ \PolyTCut{n} } := \PolyTCut{n} + \R_+^{E_n}$, which in turn is related to the dominant of the $ T $-join polytope $ \dominant{ \PolyTJoin{n} } := \PolyTJoin{n} + \R_+^{E_n}$.
We also refer to $ \dominant{ \PolyTCut{n} } $ and $ \dominant{ \PolyTJoin{n} } $ as the \emph{$ T $-cut polyhedron} and the \emph{$ T $-join polyhedron}, respectively.
The descriptions of both polyhedra in terms of linear inequalities are well-known~\cite{EdmondsJ73} (using $x(F)$ as a short-hand notation for $\sum_{e \in F} x_e$):
\begin{align}
  \label{eqDescriptionTJoin}
  \dominant{ \PolyTJoin{n} } &= \setdef{ x \in \R_+^{E_n} }[ x(C) \geq 1 \text{ for all $T$-cuts $C$ }  ] \\
  \nonumber
  \dominant{ \PolyTCut{n} } &= \setdef{ x \in \R_+^{E_n} }[ x(J) \geq 1 \text{ for all $T$-joins $J$ }  ]
\end{align}
It is worth noting that the vertices of $\dominant{ \PolyTJoin{n} }$ are the inclusion-wise minimal $ T $-joins,
i.e., those that do not contain cycles~\cite[\S 12.2]{KorteV12} and hence are edge-disjoint unions of $ \frac{1}{2}|T| $ paths whose endnodes are distinct and in $ T $.
The perfect-matching polytope $ \PolyPerfMatch{|T|} $ is a face of $ \dominant{\PolyTJoin{n}} $, induced by $ x(\delta(v)) \geq 1 $ for all $ v \in T $ and $ x(\delta(v)) \geq 0 $ for all $ v \in V_n \setminus T $.\footnote{We use the short-hand notation $ \delta(v) := \delta(\setdef{v}) $.}
Thus, from Rothvoß' proof for the exponential lower bound on the extension complexity of the perfect-matching polytope it follows that
\begin{equation}
  \label{eqIFOR1}
  \xc(\dominant{\PolyTJoin{n}}) \geq 2^{\orderOmega{|T|}}.
\end{equation}
It turns out that this bound is essentially tight.
In fact, in Appendix~\ref{SectionUpperBound} we give a linear extension for $ \dominant{\PolyTJoin{n}} $ showing
\begin{equation}
  \label{eqIFOR2}
  \xc(\dominant{\PolyTJoin{n}}) \le \orderO{ n^2 \cdot 2^{|T|} }.
\end{equation}
Thus, for case $ T = V_n $ with $ n $ even we obtain that the extension complexity of the $ V_n $-join polyhedron grows exponentially in $ n $.
In the next section we will see that this result carries over to the $ V_n $-cut polyhedron, also known as the \emph{odd-cut polyhedron}.

%
%

\section{Blocking pairs of polyhedra}
\label{SectionBlocking}

\DeclareDocumentCommand\blocker{m}{B(#1)}%

The $ T $-cut polyhedron and the $ T $-join polyhedron belong to the class of \emph{blocking polyhedra}.
A polyhedron $ P \subseteq \R^d_+ $ is \emph{blocking} if $ x' \geq x $ implies $ x' \in P $ for all $ x \in P $.
Such a polyhedron can be described as $ P = \setdef{ x \in \R_+^d }[ \scalprod{y^{(i)}}{x} \geq 1 \text{ for } i = 1,\dotsc,m ] $ for certain nonnegative vectors $y^{(1)},\dotsc,y^{(m)} \in \R_+^d$ or as $ P = \conv \setdef{x^{(1)},\dotsc,x^{(k)}} + \R_+^d $ for certain nonnegative vectors $x^{(1)}, \dotsc, x^{(k)} \in \R_+^d$.
The \emph{blocker} of $P$, defined via
\[
  \blocker{P} := \setdef{ y \in \R_+^d }[ \scalprod{x}{y} \geq 1 \text{ for all } x \in P ],
\]
is again a blocking polyhedron and satisfies $\blocker{\blocker{P}} = P$.
We refer to Section~9.2 in Schrijver's book~\cite{Schrijver86} for the proofs and more properties of blocking polyhedra.

In what follows, we will establish some connections between extension complexities of (certain faces of) blocking polyhedra and (certain faces of) their blockers.
We will make use of the following key observation of Martin~\cite{Martin91} that relates the extension complexities of certain polyhedra, in particular if they are in a blocking relation.
\begin{proposition}[{\cite{Martin91}, see also~\cite[Prop.~1]{ConfortiKWW15}}]
  \label{TheoremSeparation}
  Given a non-empty polyhedron $Q$ and $\gamma \in \R$, let
  \[
    P = \setdef{ x }[ \scalprod{y}{x} \geq \gamma \text{ for all } y \in Q ].
  \]
  Then $ \xc[P] \le \xc[Q] + 1 $.
\end{proposition}
A first consequence of Proposition~\ref{TheoremSeparation} is that the extension complexities of a blocking polyhedron $P$ and its blocker $\blocker{P}$ differ by at most $d$ (due to the nonnegativity constraints).
Thus, the extension complexities of $ \dominant{ \PolyTCut{n} } $ and $ \dominant{ \PolyTJoin{n} } $ differ by at most $ \binom{n}{2} $.
In particular, in view of~\eqref{eqIFOR1} and~\eqref{eqIFOR2}, we obtain
\begin{gather}
  2^{\orderOmega{|T|}} \le \xc(\dominant{\PolyTCut{n}}) \le \orderO{ n^2 \cdot 2^{|T|} }. \label{ExtensionComplexityTCutPoly}
\end{gather}
\DeclareDocumentCommand\polarface{mm}{F_{#1}(#2)}%
The main purpose of this section, however, is to show that a radial cone of a blocking polyhedron can be analyzed by considering a certain face of the blocker.
To this end, let us now consider a general pair $ (P, \blocker{P}) $ of blocking polyhedra in $ \R^d_+ $.
For every point $ v \in P $ we define the set
\begin{equation}
  \label{EquationPolarFaceOuter}
  \polarface{\blocker{P}}{v} := \setdef{ y \in \blocker{P} }[ \scalprod{v}{y} = 1 ]
  = \setdef{ y \in \R_+^d }[ \scalprod{ v }{ y } = 1, \, \scalprod{x}{y} \geq 1 ~\forall x \in P  ], 
\end{equation}
which is a face of $ \blocker{P} $.
The following lemma establishes structural connections between $\radialcone{P}{v}$ and $\polarface{\blocker{P}}{v}$.
\begin{lemma}
  \label{TheoremStructure}
  Let $ P \subseteq \R^d_+ $ be a blocking polyhedron and let $ v \in P $.
  \begin{enumerate}[label=(\roman*)]
  \item
    \label{TheoremPolarFaceSeparation}
    $\polarface{\blocker{P}}{v} = \setdef{ y \in \R^d }[ \scalprod{ v }{ y } = 1, ~\scalprod{x}{y} \geq 1 ~\forall x \in \radialcone{P}{v} ] $.
  \item
    \label{TheoremRadialConeSeparation}
    $ \radialcone{P}{v} = \setdef{ x \in \R^d }[ \scalprod{y}{x} \geq 1 ~\forall y \in \polarface{\blocker{P}}{v} ]$.
  \end{enumerate}
\end{lemma}
\begin{proof}
  We first prove ``$\subseteq$'' of part~\ref{TheoremPolarFaceSeparation}.
  To this end we will show that, for all $x \in \radialcone{P}{v} \setminus P$, inequality $\scalprod{ x }{ y } \geq 1$ is valid for $\polarface{\blocker{P}}{v}$.
  Let $x \in \radialcone{P}{v}$, i.e., there exist $x^{(1)}, \dotsc, x^{(k)} \in P$ and $\mu_1, \dotsc, \mu_k \geq 0$ with $x = v + \sum_{i=1}^k \mu_i (x^{(i)} - v)$.
  Then
  \begin{gather*}
    \scalprod{x}{y}
    = \scalprod{ v + \sum_{i=1}^k \mu_i (x^{(i)} - v) }{ y }
    = \scalprod{ v }{ y }  + \sum_{i=1}^k \mu_i \scalprod{ (x^{(i)} - v) }{ y }
    \geq 1
  \end{gather*}
  follows from $\scalprod{v}{y} = 1$, $\mu_i \geq 0$ and $\scalprod{ x^{(i)} }{y} \geq 1$ for all $i \in [k]$.

  To prove ``$\supseteq$'' of part~\ref{TheoremPolarFaceSeparation},
  we have to show for every $j \in [d]$ that the nonnegativity constraint $y_j \geq 0$ is redundant in the right-hand side of~\eqref{EquationPolarFaceOuter}.
  From $v + \unitvec{j} \in P$ we obtain the valid inequality $\scalprod{ v + \unitvec{j} }{ y } \geq 1$. Subtracting $\scalprod{ v }{ y } = 1$ implies the desired inequality $y_j \geq 0$.

  Before we turn to the proof of part~\ref{TheoremRadialConeSeparation}, let us fix some notation.
  Denote by
  $I := \setdef{ i \in [m] }[ \scalprod{ v }{ y^{(i)} } = 1 ]$
  and
  $J := \setdef{ j \in [d] }[ v_j = 0 ]$
  the index sets of the inequalities of $P$ that are tight at $v$.
  In other words, $\radialcone{P}{v} = \setdef{ x \in \R^d }[ \scalprod{ y^{(i)} }{ x } \geq 1 \text{ for all } i \in I \text{ and } x_j \geq 0 \text{ for all } j \in J ]$.

  To prove ``$\subseteq$'' of part~\ref{TheoremRadialConeSeparation}, we consider vectors $\hat{x} \in \radialcone{P}{v}$ and $\hat{y} \in \polarface{\blocker{P}}{v}$ and claim that $\scalprod{ \hat{x} }{ \hat{y} } \geq 1$.
  In particular, $\hat{y} \in \blocker{P}$, and hence there exists a vector $\bar{y} \leq \hat{y}$ with $\bar{y} \in \conv\{ y^{(i)} \mid i \in [m] \}$.
  
  From $\hat{y} \in \polarface{\blocker{P}}{v}$ and nonnegativity of $v$ we obtain $1 = \scalprod{v}{\hat{y}} \geq \scalprod{v}{\bar{y}}$.
  Since $\scalprod{v}{y^{(i)}} \geq 1$ holds for all $i \in [m]$, this implies $\scalprod{v}{\bar{y}} \geq 1$, and hence $\hat{y}_j = \bar{y}_j$ for all $j \in [d] \setminus J$.
  Furthermore, only $y^{(i)}$ for $i \in I$ can participate in the convex combination (of $\bar{y}$) with a strictly positive multiplier.
  Considering the inequalities that are valid for $\radialcone{P}{v}$, we observe that $\hat{x}_j \geq 0$ for all $j \in J$ and that $\scalprod{ \hat{x} }{y^{(i)}} \geq 1$ for all $i \in I$.
  This implies the desired inequality $\scalprod{ \hat{x} }{ \hat{y} } \geq \scalprod{ \hat{x} }{ \bar{y} } \geq 1$.

  It remains to prove ``$\supseteq$'' of part~\ref{TheoremRadialConeSeparation}.
  To this end, consider a vector $\hat{x}$ from the set on the right-hand side of the equation.
  For all $i \in I$, $y^{(i)} \in \polarface{\blocker{P}}{v}$ implies $\scalprod{y^{(i)}}{ \hat{x} } \geq 1$.
  Consider an arbitrary $\bar{y} \in \polarface{\blocker{P}}{v}$ and some $j \in J$.
  For all $\mu \geq 0$, we have $(\bar{y} + \mu \unitvec{j}) \in \polarface{\blocker{P}}{v}$.
  To see this, consider~\ref{EquationPolarFaceOuter} and observe that $\scalprod{v}{\unitvec{j}} = 0$ and that $\scalprod{x}{\unitvec{j}} \geq 0$ for all $x \in P$.
  In particular $1 \leq \scalprod{ \hat{x} }{ \bar{y} + \mu \unitvec{j} } = \scalprod{ \hat{x} }{ \bar{y} } + \mu \hat{x}_j$, which implies $\hat{x}_j \geq 0$ and concludes the proof.
\end{proof}
We conclude this section with the following result, which is an immediate consequence of Proposition~\ref{TheoremSeparation} and parts~\ref{TheoremPolarFaceSeparation} and~\ref{TheoremRadialConeSeparation} of Lemma~\ref{TheoremStructure}.
\begin{theorem}
  \label{TheoremRadialConeExtensionComplexity}
  Let $ P \subseteq \R_+^d$ be a blocking polyhedron and let $v \in P$.
  Then $\xc(\radialcone{P}{v})$ and $\xc( \polarface{ \blocker{P} }{v} )$ differ by at most $ 1 $.
\end{theorem}

\section{Radial cones of \texorpdfstring{$ T $}{T}-join and \texorpdfstring{$ T $}{T}-cut polyhedra}
\label{SectionMain}
In this section we will apply our structural results from the previous section to the radial cones of $T$-join and $T$-cut polyhedra.
These results relate the the extension complexities of radial cones to the extension complexities of certain faces of the blocker.
We start by reproving the result of Ventura and Eisenbrand~\cite{VenturaE03} for which we use the well-known theorem of Balas on unions of polyhedra.
\begin{proposition}[\cite{Balas74}]
  \label{TheoremBalas}
  Let $ P_1,\dotsc,P_k \subseteq \R^d $ be non-empty polyhedra, and let $ P $ be the closure of $ \conv ( P_1 \cup \dotsb \cup P_k ) $.
  Then $ \xc(P) \le \sum_{i=1}^k ( \xc(P_i) + 1 )$.
\end{proposition}

\begin{theorem}[Ventura \& Eisenbrand, 2003~\cite{VenturaE03}]
  \label{TheoremTJoinRadialCone}
  For every set $ T \subseteq V_n $ with $ |T| $ even and every vertex $ v $ of $ \dominant{\PolyTJoin{n}} $ corresponding to a $ T $-join $ J \subseteq E_n $ in $ K_n $, the extension complexity of the radial cone of $ \dominant{\PolyTJoin{n}} $ at $ v $ is most $ \orderO{ |J| \cdot n^2 } $.
\end{theorem}

The crucial observation for (re)proving the result is that the facets of the $T$-cut polyhedra have small extension complexities.

\DeclareDocumentCommand\PolySTCut{oom}{\ensuremath{P_{\text{\IfValueTF{#1}{$#1$-$#2$-cut}{$s$-$t$}-cut}}(#3)}}%

\begin{proof}
  By Theorem~\ref{TheoremRadialConeExtensionComplexity} it suffices to prove that the extension complexity of
  \[
    P := \setdef{x \in \dominant{\PolyTCut{n}}}[\scalprod{v}{x} = 1]
  \]
  is at most $ \orderO{ |J| \cdot n^2} $.
  A vector $y \in \R^{E_n}$ is in the recession cone $ C $ of $P$ if and only if it is nonnegative and $\scalprod{v}{y} = 0$ holds.
  Thus, $ C $ is generated by all unit vectors corresponding to edges in $ E_n \setminus J $.
  For every edge $ m \in J $ we consider the set
  \[
    F_m := \setdef{ x \in P }[ x_e = 0 ~\forall e \in J \setminus \setdef{m} ],
  \]
  which is a face of $ P $.
  Note that since $ \scalprod{v}{x} = 1$ is valid for $ F $, so is $ x_m = 1 $.
  It is easy to see that $F_m$ also has $ C $ as its recession cone.
  Every vertex $ w $ of $ P $ satisfies $w_m = 1$ for some edge $ m \in J $, and thus $ w \in F_m $, which (since $ P $ and all faces $ F_m $ have the same recession cone) proves
  \[
    P = \conv( \bigcup_{m \in J} F_m ).
  \]
  Hence, by Proposition~\ref{TheoremBalas}, $ \xc(P) \leq |J| \cdot ( \xc(F_m) + 1 ) $ holds, and it remains to prove $ \xc(F_m) \leq \orderO{ n^2 } $ for all $ m \in J $.
  We claim that $ F_m $ is equal to
  \[
    G_m := \setdef{ x \in \dominant{\PolyTCut[T']{n}} }[ x_m = 1, ~ x_e = 0 ~\forall e \in J \setminus \setdef{m} ],
  \]
  where $ T' := m $ is the set containing the two endnodes of $ m $.
  Note that $G_m$ is a face of $ \dominant{\PolyTCut[T']{n}} $ and hence both polyhedra are integral.
  Moreover, $ G_m $ also has $ C $ as its recession cone.
  To see that also their vertex sets agree, consider a cut $ \delta(S) $ for some $ S \subseteq V $.
  If $ \delta(S) $ is a $ T $-cut that contains $ m $, then $ \delta(S) $ is also a $ T' $-cut.
  Suppose $ \delta(S) $ is a $ T' $-cut with $ \delta(S) \cap J = \setdef{m} $.
  Since $ J $ is the edge-disjoint union of paths whose endnodes are distinct and in $ T $, all such paths, except for the one that contains edge $ m $, have both endnodes either in $ S $ or in $ V_n \setminus S$.
  This shows that $ |S \cap T| $ is odd and hence that $ \delta(S) $ is a $ T $-cut.
  This concludes the proof of the claim that $F_m = G_m$ holds.

  Since $T'$ contains exactly two nodes, Proposition~\ref{TheoremExtensionBasics}~\ref{TheoremExtensionBasicsFace} and the upper bound from~\eqref{ExtensionComplexityTCutPoly} already yield $\xc(G_m) \leq \xc(\PolyTCut[T']{n}) \leq \orderO{ n^2 }$, which concludes the proof.
\end{proof}

\begin{corollary}[Proposition~2.1 in~Ventura \& Eisenbrand, 2003~\cite{VenturaE03}]
  For every $ n $ and every vertex $ v $ of $ \PolyPerfMatch{n} $, the extension complexity of the radial cone of $ \PolyPerfMatch{n} $ at $ v $ is most $ \orderO{ n^3 } $.
\end{corollary}

\begin{proof}
  The result follows from Theorem~\ref{TheoremTJoinRadialCone} and Proposition~\ref{TheoremExtensionBasics}~\ref{TheoremExtensionBasicsRadialConeFace}, using the fact that
  $ \PolyPerfMatch{n} $ is a face of $ \dominant{\PolyTJoin{n}} $ (see Section~\ref{SectionOverview}).
  Note that the bound is cubic since $ v $ corresponds to a perfect matching, which consists of $ n / 2 $ edges.
\end{proof}

\begin{corollary}
  For every $ n $ and every $ v \in \dominant{\PolyTJoin{n}} $, the extension complexity of the radial cone of $ \dominant{\PolyTJoin{n}} $ at $ v $ is most $ \orderO{ n^4 } $.
\end{corollary}

\begin{proof}
  Let $ P := \dominant{ \PolyTJoin{n} }$ and let $ w $ be a vertex of $ P $ in the smallest face that contains $ v $.
  Theorem~\ref{TheoremTJoinRadialCone} implies that the extension complexity of $ \radialcone{ P }{ w } $ ist at most $ \orderO{ n^4 } $.
  By definition of the radial cone, $\radialcone{ P }{ w } \subseteq \radialcone{ P }{ v } $,
  and thus, by Lemma~\ref{TheoremStructure}, $ \polarface{\blocker{P}}{v} \subseteq \polarface{\blocker{P}}{w} $.
  Using the fact that $ \polarface{\blocker{P}}{v} $ and $ \polarface{\blocker{P}}{w} $ are faces of $\blocker{P}$,
  this implies that $ \polarface{\blocker{P}}{v} $ is a face of $ \polarface{\blocker{P}}{w} $.
  Theorem~\ref{TheoremRadialConeExtensionComplexity} and Proposition~\ref{TheoremExtensionBasics}~\ref{TheoremExtensionBasicsFace} yield
  \[
    \xc( \radialcone{ P }{ v } ) 
    \leq \xc( \polarface{\blocker{P}}{v} ) + 1
    \leq \xc ( \polarface{\blocker{P}{w}} ) + 1
    \leq \xc( \radialcone{ P }{ w } ) + 2
    \leq \orderO{ n^4 },
  \]
  which concludes the proof.
\end{proof}

We continue with the main result of this paper.
To prove it, we again relate the extension complexity of the radial cones to the extension complexities of certain faces of the blocker, i.e., the $T$-join polyhedron.
In contrast to the situation for Theorem~\ref{TheoremTJoinRadialCone}, these faces are again very related to $T$-join polyhedra, and thus have high extension complexities.

\begin{theorem}
  \label{TheoremMain}
  For sets $ T \subseteq V_n $ with $ |T| $ even and vertices $ v $ of $ \dominant{\PolyTCut{n}} $,
  the extension complexity of the radial cone of $ \dominant{\PolyTCut{n}} $ at $ v $ is at least $2^{\orderOmega{|T|}}$.
\end{theorem}

\begin{proof}
  By Theorem~\ref{TheoremRadialConeExtensionComplexity} it suffices to prove that the extension complexity of
  \[
    P := \setdef{x \in \dominant{\PolyTJoin{n}}}[\scalprod{v}{x} = 1]
  \]
  is at least $2^{\orderOmega{|T|}}$.
  To this end, we will construct a face $ Q $ of $ P $ that is a Cartesian product of a $ T_1 $-join polyehdron, a single point, and a $ T_2 $-join polyhedron for some $ T_1, T_2 \subseteq T $ with $ |T_1| + |T_2| + 2 = |T| $.
  Note that, by Proposition~\ref{TheoremExtensionBasics} and Inequality~\eqref{eqIFOR1}, this will imply
  \[
    \xc(P) \ge \xc(Q) \ge \max \left( \xc(\PolyTJoin[T_1]{n_1}), \xc(\PolyTJoin[T_2]{n_2}) \right) \ge 2^{\orderOmega{|T|}}.
  \]
  For subsets $ V_1, V_2 \subseteq V $, we will use the notation $ V_1 : V_2 := \setdef{\{v_1,v_2\}}[v_1 \in V_1, \, v_2 \in V_2] $ as well as $ E(V_1) := \setdef{\{v,w\}}[v,w \in V_1, \, v \ne w] $.
  Recall that $ v \in \R^{E_n} $ is a vertex of $ \dominant{\PolyTCut{n}} $ and hence we can partition $ V $ into sets $ U_1,U_2 $ with $ |T \cap U_1| $ odd and $ |T \cap U_2| $ odd, such that $ v $ is the characteristic vector of $ U_1:U_2 $.
  With this notation the set $ P $ can be rewritten as
  \[
    P = \setdef{x \in \dominant{\PolyTJoin{n}}}[x(U_1:U_2) = 1].
  \]
  Fix $ t_1 \in T \cap U_1 $ and $ t_2 \in T \cap U_2 $, and define
  \begin{align*}
    V_i & := U_i \setminus \{t_i\}, \\
    T_i & := (T \cap U_i) \setminus \{t_i\} \quad i=1,2.
  \end{align*}
  Let 
  \[
    F := (V_1:V_2) \cup (V_1 : \{t_1,t_2\}) \cup (V_2 : \{t_1,t_2\})
  \]
  denote the set of edges that lie between (any two of) the three sets $ V_1 $, $ V_2 $, and $ \{t_1,t_2\} $, and consider the set
  \[
    Q := \setdef{x \in P}[x_e = 0 \text{ for all } e \in F],
  \]
  which is a face of $ P $.
  The support of each point $ x \in Q $ is contained in $ E(V_1) \cup E(V_2) \cup \setdef{\{t_1,t_2\}} $.
  Furthermore, for each $ x \in Q $ we have
  \[
    x_{\{t_1,t_2\}}
    = \underbrace{x(V_1:V_2)}_{=0} + \underbrace{x(\{t_1\}:V_2)}_{=0} + \underbrace{x(V_1:\{t_2\})}_{=0} + x_{\{t_1,t_2\}}
    = x(U_1:U_2)
    = 1,
  \]
  and hence $ Q = \setdef{x \in \dominant{\PolyTJoin{n}}}[x_e = 0 \text{ for all } e \in F, \, x_{\{t_1,t_2\}} = 1] $.
  By~\eqref{eqDescriptionTJoin}, we thus obtain
  \begin{alignat}{10}
    \label{eqQ1}
    Q = \Big\{ x \in \R^{E_n}_+ : & & x_e & = 1 \text{ for } e = \{t_1,t_2\}, \\
    \label{eqQ2}
                              & &            x_e & =   0 \text{ for all } e \in F, \\
    \label{eqQ3}
                              & &   x(S : (V \setminus S)) & \ge 1 \text{ for all } S \subseteq V \text{ with } |T \cap S| \text{ odd } \Big\}.
  \end{alignat}
  We claim that $ Q $ is equal to
  \begin{alignat}{10}
    \label{eqtQ1}
    \tilde{Q} = \Big\{ x \in \R^{E_n}_+ : & & x_e & = 1 \text{ for } e = \{t_1,t_2\}, \\
    \label{eqtQ2}
                              & &            x_e & =   0 \text{ for all } e \in F, \\
    \label{eqtQ3}
                              & &   x(S_1 : (V_1 \setminus S_1)) & \ge 1 \text{ for all } S_1 \subseteq V_1 \text{ with } |T_1 \cap S_1| \text{ odd}, \\
    \label{eqtQ4}
                              & &   x(S_2 : (V_2 \setminus S_2)) & \ge 1 \text{ for all } S_2 \subseteq V_2 \text{ with } |T_2 \cap S_2| \text{ odd } \Big\}.
  \end{alignat}
  Note that this establishes our main claim since, by~\eqref{eqDescriptionTJoin}, $ \tilde{Q} $ is the Cartesian product of a $ T_1 $-join polyhedron (with respect to the complete graph formed by the nodes of $ V_1 $), a $ T_2 $-join polyhedron (with respect to the complete graph formed by the nodes of $ V_2 $), and a set consisting of a single vector in $ \R^{F \cup \{t_1,t_2\}} $ (defined by~\eqref{eqtQ1} and~\eqref{eqtQ2}).

  To this end, first note that the constraints in~\eqref{eqQ1} and~\eqref{eqQ2} are identical to~\eqref{eqtQ1} and~\eqref{eqtQ2}.
  To see that $ Q \subseteq \tilde{Q} $, let $ x \in Q $, $ i \in \{1,2\} $, and $ S_i \subseteq V_i $ with $ |T_i \cap S_i| $ odd.
  By~\eqref{eqQ2} we have $ x(S_i : (V_i \setminus S_i)) = x(S_i : (V \setminus S_i)) $.
  Since $ |S_i \cap T_i| = |S_i \cap T| $ is odd, by~\eqref{eqQ3} we also have $ x(S_i : (V \setminus S_i)) \ge 1 $, which shows that the constraints in~\eqref{eqtQ3} and~\eqref{eqtQ4} are satisfied and hence $ x \in \tilde{Q} $.

  To see that $ \tilde{Q} \subseteq Q $, let $ x \in \tilde{Q} $ and $ S \subseteq V $ with $ |T \cap S| $ odd.
  If $ |S \cap \{t_1,t_2\}| = 1 $, then the nonnegativity of $ x $ and~\eqref{eqtQ1} already imply
  \[
    x(S : (V \setminus S)) \ge x_{\{t_1,t_2\}} = 1.
  \]
  Otherwise, we have $ |S \cap \{t_1,t_2\}| \in \{0,2\} $, define $ S_i := S \cap V_i $ for $ i=1,2 $.
  Since $ |S \cap T| = |S \cap \{t_1,t_2\}| + |S_1 \cap T_2| + |S_2 \cap T_2| $ is odd, we must have that $ |S_{i^*} \cap T_{i^*}| $ is also odd for some $ i^* \in \{1,2\} $.
  By the constraints in~\eqref{eqtQ3} and~\eqref{eqtQ4}, this implies $ x(S_{i^*} : (V_{i^*} \setminus S_{i^*})) \ge 1 $.
  We finally obtain
  \[ 
    x(S : (V \setminus S)) = x(S_1 : (V_1 \setminus S_1)) + x(S_2 : (V_2 \setminus S_2))
    \ge x(S_{i^*} : (V_{i^*} \setminus S_{i^*})) \ge 1,
  \]
  where the equality follows from~\eqref{eqtQ2}, and the first inequality is due to nonnegativity of $ x $.
\end{proof}

Notice that from Theorem~\ref{TheoremMain} we obtain Theorem~\ref{TheoremMainIntro} by choosing $T := V_n$.

\bibliographystyle{plain}
\bibliography{references}

\appendix

\section{Upper bound for small cardinalities}
\label{SectionUpperBound}

In this section we establish an upper bound of $ \orderO{ n^2 \cdot 2^{|T|} } $
on the extension complexities of $ T $-join- and $ T $-cut polyhedra.

\begin{lemma}
  \label{TheoremDominantTJoinExtension}
  For every $ n $ and every set $ T \subseteq V $, the extension complexity of $ \dominant{\PolyTJoin{n}} $ is bounded by $ \orderO{n^2 \cdot 2^{|T|}} $.
\end{lemma}

\DeclareDocumentCommand\outArcs{m}{\delta^{\text{out}}\left(#1\right)}%
\DeclareDocumentCommand\inArcs{m}{\delta^{\text{in}}\left(#1\right)}%

For every node $ v $ of a directed graphs we denote by $\outArcs{v}$ and $\inArcs{v}$ the sets of arcs leaving (resp.\ entering) $ v $.

\begin{proof}
  Trivially, we only have to consider the case of $|T|$ even, since $\PolyTJoin{n} = \emptyset$ otherwise.
  Let $A := \setdef{ (u,v),(v,u) }[ \setdef{u,v} \in E_n ]$ denote the set of bidirected edges of $K_n$.
  We define, for $S \subseteq T$ with $|S| = |T| / 2$ the polyhedron
  \begin{multline}
    \label{UpperBoundExtension}
    P_S := \{ x \in \R^{E_n} : \exists f \in \R_+^A :
      f(\outArcs{v}) - f(\inArcs{v}) = 
      \begin{cases}
        1       & \text{ if $v \in S$} \\ 
        -1      & \text{ if $v \in T \setminus S$} \\ 
        0       & \text{ if $v \in V \setminus T$} 
      \end{cases} \\
      \text{for all $ v \in V $ and $ x_{\setdef{u,v}} \geq f_{(u,v)} + f_{(v,u)} $ for all $ \setdef{u,v} \in E_n $} \}.
  \end{multline}
  It is easy to see that the extension of $ P_S $ is an integer polyhedron since the first set of constraints defines a totally unimodular system with integral right-hand side, and since every $ x $-variable appears in only one of the further inequalities.
  Clearly, $P_S$ is an integer polyhedron as well, since the projection on the $x$-variables maintains integrality.

  We claim that $ P_S \subseteq \dominant{\PolyTJoin{n}} $ holds.
  To this end, let $x \in P_S$ and let $f \in \R_+^A$ be such that the constraints in~\eqref{UpperBoundExtension} are satisfied.
  For each node $v \in T$, we obtain $x(\delta(v)) \geq \sum_{\setdef{u,v} \in \delta(v)} (f_{(u,v)} + f_{(v,u)}) \geq 1$.
  By integrality of $P_S$, this suffices to prove the claim.
  Let now $J$ be a $T$-join.
  It is an edge-disjoint union of circuits $C_1, \dotsc, C_k$ and paths $P_1, \dotsc, P_\ell$ for $\ell = \frac{1}{2}|T|$ connecting disjoint pairs of nodes in $T$.
  For $i \in [k]$, let $\vec{C}_i \subseteq A$ be a directed version of $C_i$, that is, a directed cycle whose underlying undirected cycle is $C_i$.
  For $j \in [\ell]$, let $\vec{P}_j \subseteq A$ be a directed version of $P_j$, that is, a directed path whose underlying undirected path is $P_j$.
  Let $S \subseteq T$ be the set of starting nodes of the paths $\vec{P}_j$.
  Define $x := \chi(J)$ and for all $a \in A$, $f_a := 1$ if $a \in \vec{C}_i$ for some $i \in [k]$ or $a \in \vec{P}_j$ for some $j \in [\ell]$, and $f_a := 0$ otherwise.
  By construction, $(x,f)$ satisfies the constraints in~\eqref{UpperBoundExtension}, which shows $x \in P_S$.

  This proves that the vertex set of $\dominant{\PolyTJoin{n}}$ is covered by the union of the polyhedra $P_S$ for all $S \subseteq T$ with $|S| = \frac{1}{2} |T|$.
  Proposition~\eqref{TheoremBalas} yields desired result since there are less than $2^{|T|}$ such sets $S$ and $\xc(P_S) \leq 3|E_n|$ holds.
\end{proof}

\begin{corollary}
  \label{TheoremDominantTCutExtension}
  For every $ n $ and every set $ T \subseteq V $, the extension complexity of $ \dominant{\PolyTCut{n}} $ is bounded by $ \orderO{n^2 \cdot 2^{|T|}} $.
\end{corollary}

\begin{proof}
  Apply Proposition~\ref{TheoremSeparation} to Lemma~\ref{TheoremDominantTJoinExtension}.
\end{proof}

\end{document}